\documentclass[%
 reprint,
 amsmath,amssymb,
aip,jcp]{revtex4-2}

\usepackage{graphicx}
\usepackage{dcolumn}
\usepackage{bm}
\usepackage{hyperref}

\usepackage{tikz}
\usetikzlibrary{arrows.meta}

\usepackage{amsthm}

\newtheorem{theorem}{Theorem}

\theoremstyle{definition}

\theoremstyle{remark}

\usepackage[T1]{fontenc}

\usepackage{lineno}

\begin{document}


\title{Geometric Thermodynamics of Cycles: Curvature and Local Thermodynamic Response}

\author{Eric R. Bittner}
\affiliation{Department of Physics, University of Houston}

\date{\today}

\begin{abstract}
Classical thermodynamics contains familiar geometric relations associated with cyclic processes, most notably the identification of mechanical work with the area enclosed by a trajectory in the $(P,V)$ plane. We show that the area laws for work and reversible heat arise as projections of a single canonical two--form defined on the equilibrium thermodynamic manifold, providing a unified description of thermodynamic cycles in both the $(P,V)$ and $(T,S)$ representations. The same structure yields a direct link between cycle geometry and thermodynamic response: the work generated by infinitesimal cycles is set locally by the mixed curvature $U_{SV}$ of the equilibrium energy surface, which can be expressed in terms of measurable susceptibilities. This identifies thermodynamic work as a local geometric field over state space rather than solely a global property of cyclic processes. More broadly, the framework connects classical cycle geometry to stochastic thermodynamic trajectories, providing a geometric interpretation of nonequilibrium work relations such as the Jarzynski equality.
\end{abstract}
\maketitle

\section{Introduction}

Many central relations of thermodynamics admit a geometric
interpretation as statements about curves and surfaces in state space.
The most familiar example is the relation between mechanical work and
the area enclosed by a cycle in the $(P,V)$ plane,
\begin{equation}
W = \oint P\,dV ,
\end{equation}
which identifies the work performed by the system during a thermodynamic
cycle with the area enclosed by the corresponding trajectory.\footnote{Throughout
this work we consider quasistatic transformations and adopt the convention
that work is done by the system, so that the reversible work differential
takes the form $\delta W = P\,dV$. All cyclic integrals are taken with
counterclockwise (positive) orientation in the relevant coordinate plane.}
In the $(T,S)$ representation, the local geometry of a thermodynamic
trajectory provides a complementary interpretation: the tangent vector
decomposes the motion into adiabatic and heat-exchanging components. Its
projection along the entropy axis determines the reversible heat flow
$\delta Q = T\,dS$, while motion along the temperature axis corresponds
to adiabatic variation. These constructions are standard elements of
thermodynamic analysis and appear in many treatments of the subject,
notably in the classic text of Callen~\cite{Callen1985}. They are also
naturally expressed in the language of differential forms: work and heat
appear as line integrals over closed curves, while the associated area
laws follow from the corresponding surface integrals in state space.
This geometric viewpoint has been emphasized in modern treatments of
differential forms, such as the work of Needham~\cite{Needham2021}, where
line and surface integrals are understood directly as oriented geometric
objects.

Despite the familiarity of these constructions, the underlying
geometric structure of classical thermodynamics is seldom stated
explicitly. The First Law,
\begin{equation}
dU = T\,dS - P\,dV ,
\end{equation}
defines a differential relation among the thermodynamic variables
$(U,S,V,T,P)$ and thereby endows thermodynamic state space with a
natural geometric structure. The familiar area relations of
thermodynamic cycles arise as projections of this structure onto
specific coordinate planes. Figure~\ref{fig:thermo-geometry}
illustrates this correspondence: a cycle in the $(P,V)$ plane maps to
a corresponding curve in the $(T,S)$ plane, representing the same
thermodynamic process in a different coordinate chart. The two area
laws are therefore not independent constructions, but different
representations of the same underlying geometric content.

\begin{figure*}[t]
\centering
\includegraphics[width=0.95\textwidth]{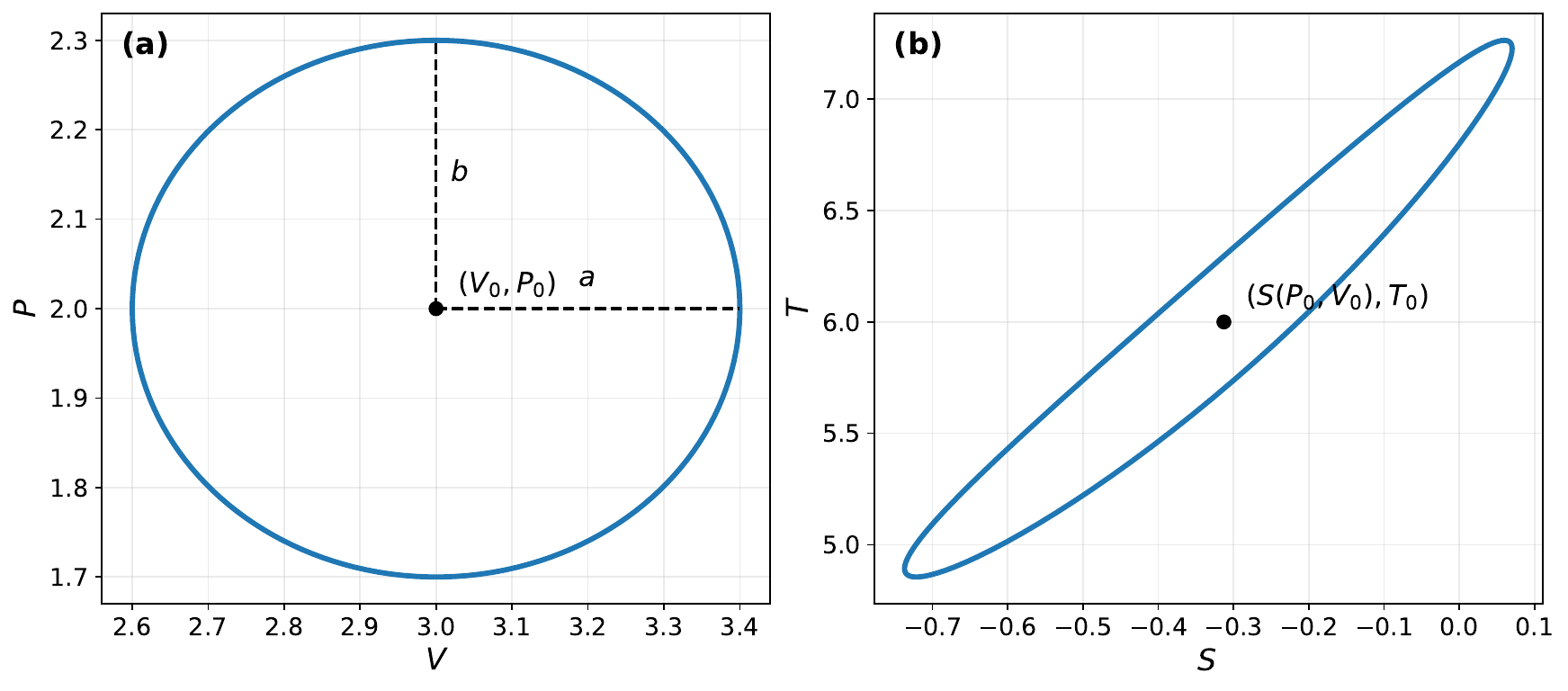}
\caption{
\textbf{Local mapping of a quasistatic cycle from the $(P,V)$ plane to the $(T,S)$ plane for an ideal gas.}
(a) An elliptical cycle in the $(P,V)$ plane centered at $(V_0,P_0)$ with
semi-axes $a$ and $b$.
(b) The corresponding image of the same cycle in the $(T,S)$ plane under the
exact ideal-gas mapping. The black points indicate the reference state
$(V_0,P_0)$ and its image $(S(P_0,V_0),T_0)$. The nonlinear transformation
from $(P,V)$ to $(T,S)$ does not preserve geometric centers, so the image of
the center state need not coincide with the geometric center of the mapped
curve. Parameters used are $P_0=2$, $V_0=3$, $a=0.4$, $b=0.3$, $n=1$,
$C_V=3/2$, and units are chosen such that $R=1$.
}
\label{fig:thermo-geometry}
\end{figure*}

Geometric approaches to thermodynamics have developed along several
lines. Weinhold and Ruppeiner introduced metric structures on
thermodynamic state space and related them to equilibrium fluctuations
and phase behavior~\cite{Weinhold1975,Ruppeiner1995}. A complementary
approach based on contact geometry treats the First Law as a contact
one-form and identifies equilibrium states with Legendre submanifolds
of thermodynamic phase space~\cite{Hermann1973,Mrugala1991,Mrugala2000,Bravetti2019,vanderSchaft2018}. Geometric ideas also appear in
nonequilibrium statistical mechanics through fluctuation relations such
as the Jarzynski equality and the Crooks theorem, which relate
stochastic work to equilibrium free-energy
differences~\cite{Jarzynski1997,Crooks1999}, while extensions of
contact geometry to mesoscopic and nonequilibrium dynamics provide a
framework for irreversible evolution and model
reduction~\cite{Grmela2014}. Although the present analysis is formulated
within classical equilibrium thermodynamics, it connects naturally to
modern treatments of work and heat as trajectory-dependent quantities
in driven systems and stochastic
dynamics~\cite{Jarzynski1997,Crooks1999,Seifert2012}. From this
perspective, classical cycle relations and nonequilibrium work
functionals appear as different manifestations of a common geometric
structure.

The present work is closest in spirit to the contact-geometric
viewpoint, but emphasizes a complementary and more concrete result: the
area laws for work and reversible heat arise as symplectic projections
of the underlying contact structure. In particular, we show that these
relations follow from a single canonical two-form defined on the
equilibrium thermodynamic manifold, and that the same object gives a
local description of thermodynamic response. In the energy
representation $U(S,V)$, the mixed derivative $U_{SV}$ plays the role
of a curvature density controlling the work generated by infinitesimal
cycles, and this quantity may be expressed directly in terms of
measurable susceptibilities. The analysis therefore identifies
thermodynamic work not only as a global property of cyclic processes
but as a local geometric field defined over equilibrium state space.

The paper develops this argument in three steps. We first review the
geometric interpretation of thermodynamic cycles in the $(P,V)$ and
$(T,S)$ planes and show how the two descriptions are related. We then
derive the corresponding area laws from the contact structure defined
by the First Law and from the geometry of the equilibrium surface
$U(S,V)$, establishing that work and reversible heat are projections of
a single canonical two-form. As shown in
Sec.~\ref{sec:response_geometry}, this structure also yields a direct
connection between cycle geometry and measurable response functions.
Finally, we connect the same geometric framework to fluctuation
relations, where thermodynamic work appears as a stochastic functional
of system trajectories. The aim is therefore not merely to restate
familiar area laws in new notation, but to identify the intrinsic
geometric structure that unifies classical cycle analysis,
thermodynamic response, and nonequilibrium work relations.
\section{Theoretical Framework}
\label{sec:II}

We now develop an explicit geometric formulation of thermodynamic cycles.
Classical cycle descriptions in different coordinate representations,
such as $(P,V)$ and $(T,S)$, are understood as coordinate charts on a
common equilibrium manifold defined by the First Law. Within this
framework, work and reversible heat emerge as projections of a single
underlying geometric structure. We show that both are governed by a
canonical two-form on the equilibrium manifold, whose flux determines
the work generated by cyclic processes. This formulation unifies the
standard area laws of thermodynamics and provides a direct link between
cycle geometry and thermodynamic response.

\subsection{Contact and Symplectic Structure of Thermodynamic Cycles}
\label{sec:IIa}

Classical thermodynamics admits a natural geometric formulation in terms of
contact manifolds\cite{Grmela2014}. For a simple compressible system, the
thermodynamic state space is described by coordinates
\[
(U,S,V,T,P),
\]
with contact one--form
\begin{equation}
\alpha = dU - T\,dS + P\,dV .
\end{equation}

Equilibrium states lie on Legendre submanifolds defined by the contact
condition $\alpha = 0$, which gives the First Law in differential form,
\begin{equation}
dU = T\,dS - P\,dV .
\end{equation}
In the energy representation, these states define a surface $U(S,V)$
embedded in $(S,V,U)$ space. The intensive variables follow from its
slopes,
\[
T = \left(\frac{\partial U}{\partial S}\right)_V,
\qquad
P = -\left(\frac{\partial U}{\partial V}\right)_S ,
\]
so that the tangent plane satisfies the First Law. A thermodynamic
process traces a trajectory on this surface; a cycle is a closed curve,
as shown in Fig.~\ref{fig:energy-surface}.

\begin{figure}[t]
\includegraphics[width=\columnwidth]{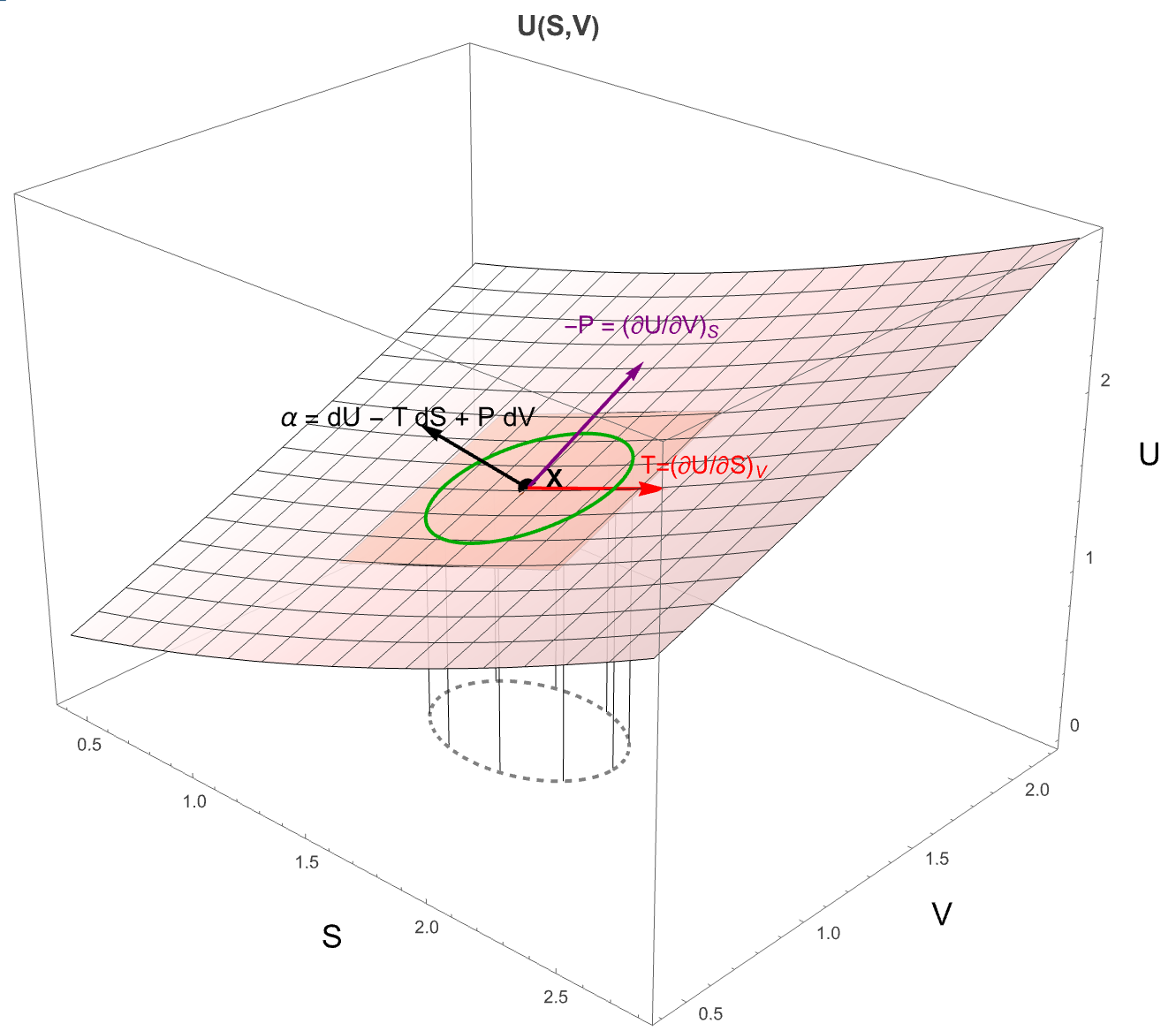}
\caption{
Geometric structure of thermodynamics in the energy representation.
The equilibrium states form a surface $U(S,V)$, whose tangent plane
encodes the First Law. Thermodynamic cycles correspond to closed curves
on this surface, and their projections onto lower-dimensional planes
define the geometric structure of work and heat. The projected area on
the $SV$ plane corresponds to the canonical two-form
$\Omega = -U_{SV} dS \wedge dV$, which gives the geometric origin of
thermodynamic work relations.
}
\label{fig:energy-surface}
\end{figure}

The full thermodynamic structure is contact, but many familiar relations
arise from projections onto two--dimensional subspaces that carry a
natural symplectic structure. In the $(P,V)$ and $(T,S)$ representations,
the two--forms
\begin{equation}
\omega_{PV} = dP \wedge dV,
\qquad
\omega_{TS} = dT \wedge dS
\end{equation}
define the area elements associated with mechanical work and reversible
heat exchange.

For a quasistatic cycle $\gamma$ in the $(P,V)$ plane, Stokes' theorem gives
\begin{equation}
W = \oint_\gamma P\,dV = \iint_{\Sigma} dP \wedge dV ,
\end{equation}
so that the work equals the symplectic area enclosed by the trajectory.
In the $(T,S)$ representation, the heat exchanged along isothermal
segments is $Q = \int T\,dS$, and reversible cycles appear as rectangular
regions whose area gives the heat transfer. These relations follow as
projections of the underlying contact structure onto symplectic
subspaces, where area integrals represent fluxes of differential
two--forms~\cite{Needham2021}.
\subsection{Local Geometry of Heat Flow}\label{sec:IIb}

The symplectic area relation for work arises from the projection of the
thermodynamic contact manifold onto the $(P,V)$ plane. To illustrate
this structure, consider a small cycle parameterized by
\begin{align}
V &= V_0 + a \cos\theta , \\
P &= P_0 + b \sin\theta .
\end{align}
The work performed over one cycle is
\begin{equation}
W = \oint P\,dV = \pi ab ,
\end{equation}
which equals the area enclosed by the ellipse. This is the local
symplectic area associated with the projection onto the $(P,V)$ plane.

A complementary structure appears in the heat differential. For a simple
compressible system, the First Law gives
\begin{equation}
\delta Q = dU + P\,dV .
\end{equation}
For an ideal gas, $U$ depends only on temperature, and using the equation
of state $PV = nRT$ yields
\begin{equation}
\delta Q = \frac{C_P}{R}P\,dV + \frac{C_V}{R}V\,dP .
\end{equation}

Near a reference state $(P_0,V_0)$ this expression linearizes to
\begin{equation}
\delta Q \approx A\,dV + B\,dP
\label{eq:11}
\end{equation}
with
\[
A = \frac{C_P}{R}P_0,
\qquad
B = \frac{C_V}{R}V_0 .
\]
This form defines a preferred local direction in the $(P,V)$ plane.
Along trajectories satisfying
\begin{equation}
A\,dV + B\,dP > 0 ,
\end{equation}
heat flows into the system, while the orthogonal direction
\begin{equation}
A\,dV + B\,dP = 0
\end{equation}
defines local adiabatic motion. The linearized heat form therefore
selects a pair of directions: one associated with heat flow and the
other with adiabatic trajectories. This structure is shown in
Fig.~\ref{fig:thermo-geometry}(a).

This construction defines a local change of coordinates in state space.
One coordinate aligns with heat flow, while the conjugate direction
aligns with adiabatic motion. In these coordinates, the heat differential
takes the canonical form
\begin{equation}
\delta Q = T\,dS .
\end{equation}
Entropy therefore appears as the coordinate associated with heat flow,
with temperature as the corresponding integrating factor.

The mapping from $(P,V)$ to $(T,S)$ may be made explicit for the same
small cycle. Linearizing about $(P_0,V_0)$ gives
\begin{equation}
T(\theta)\approx T_0+\frac{P_0 a}{nR}\cos\theta+\frac{V_0 b}{nR}\sin\theta ,
\end{equation}
and
\begin{equation}
S(\theta)\approx S_0+\frac{nC_P a}{V_0}\cos\theta+\frac{nC_V b}{P_0}\sin\theta .
\end{equation}
An elliptical cycle in the $(P,V)$ plane therefore maps to a rotated
ellipse in the $(T,S)$ plane, as shown in
Fig.~\ref{fig:thermo-geometry}(b). The two curves represent the same
thermodynamic cycle in different local coordinate charts.

The image of a small cycle in the $(T,S)$ plane is generically a smooth
closed curve determined by the local response coefficients. The familiar
rectangular geometry arises only when the cycle is composed of exact
isothermal and adiabatic segments, as in the Carnot construction. We
turn to this case in the next section.
\subsection{Geometric Origin of Thermodynamic Area Laws}
\label{sec:IIe}

The classical area relations for work and reversible heat are typically
introduced in different coordinate representations. Mechanical work is
expressed as an area in the $(P,V)$ plane, while reversible heat is
expressed as an area in the $(T,S)$ plane. Although these constructions are closely related through the structure
of equilibrium thermodynamics, their common geometric origin is rarely
formulated in terms of a single intrinsic object. On the equilibrium
manifold, both area laws arise from a canonical two--form, as made
precise below.

\begin{theorem}[Canonical thermodynamic two--form]
\label{theorem:1}
Let $(\mathcal T,\alpha)$ be the thermodynamic contact manifold of a
simple compressible system with contact form
\[
\alpha = dU - T\,dS + P\,dV .
\]
Let $\iota:\mathcal E \hookrightarrow \mathcal T$ denote the inclusion
of the equilibrium Legendre submanifold $\mathcal E$ satisfying
$\alpha = 0$. Then
\[
\iota^*(dP\wedge dV)=\iota^*(dT\wedge dS)=:\Omega .
\]
In the energy representation $U=U(S,V)$,
\[
\Omega = -U_{SV}\, dS\wedge dV,
\]
where
\[
U_{SV} \equiv \frac{\partial^2 U}{\partial S\,\partial V},
\]
denotes the mixed second derivative of the energy surface.
For any reversible cycle $C$ bounding a surface $\Sigma \subset \mathcal E$,
\[
\oint_C P\,dV
=
\int_\Sigma \Omega
=
\oint_C T\,dS .
\]
\end{theorem}

This result identifies thermodynamic work as the flux of a curvature
field on the equilibrium manifold, thereby elevating cycle work from a
global area property to a local geometric density.

\begin{proof}
On $\mathcal E$, $\iota^*\alpha=0$, so $\iota^*(d\alpha)=0$. Since
\[
d\alpha = -dT\wedge dS + dP\wedge dV ,
\]
it follows that
\[
\iota^*(dP\wedge dV)=\iota^*(dT\wedge dS).
\]

In the energy representation,
\[
T = U_S, \qquad P = -U_V ,
\]
so that
\[
dT = U_{SS}\,dS + U_{SV}\,dV,
\qquad
dP = -U_{SV}\,dS - U_{VV}\,dV .
\]
Taking exterior products gives
\[
dT\wedge dS = -U_{SV}\,dS\wedge dV,
\qquad
dP\wedge dV = -U_{SV}\,dS\wedge dV .
\]
Thus both pullbacks reduce to
\[
\Omega = -U_{SV}\,dS\wedge dV ,
\]
and the cycle relation follows from Stokes' theorem.
\end{proof}

The theorem shows that the two apparent area laws are not independent,
but represent projections of a single intrinsic two--form. The
distinction between the $(P,V)$ and $(T,S)$ representations therefore
reflects only the choice of coordinates on the equilibrium manifold.

In local coordinates $(S,V)$,
\[
\Omega = -U_{SV}\,dS\wedge dV ,
\]
so the mixed derivative $U_{SV}$ defines a curvature density governing
the work generated by infinitesimal cycles. For a small reversible loop
spanning $\delta S\,\delta V$,
\[
\delta W \approx -U_{SV}\,\delta S\,\delta V ,
\]
showing that the magnitude of $U_{SV}$ controls the local work response.
Regions where $|U_{SV}|$ is large produce stronger work, while points
where $U_{SV}=0$ correspond to local decoupling of entropy and volume.

The equality of the projected area forms encodes the Maxwell relation
\[
\left(\frac{\partial T}{\partial V}\right)_S
=
-\left(\frac{\partial P}{\partial S}\right)_V ,
\]
which expresses the integrability of the underlying structure. A
thermodynamic cycle is therefore a closed curve on $U(S,V)$ whose
projected area is measured by $\Omega$, linking the geometry of cycles
directly to thermodynamic response.
\subsection{Fluctuating Thermodynamic Trajectories}
\label{sec:IIf}

In driven systems, thermodynamic trajectories fluctuate due to
microscopic degrees of freedom, and the work performed during a process
becomes a stochastic quantity.

Consider a protocol in which an external parameter $\lambda(t)$ drives
the system between equilibrium states. The work along a trajectory
$\gamma$ is the path functional
\begin{equation}
W[\gamma] = \int_0^\tau
\frac{\partial H}{\partial \lambda}\,
\dot{\lambda}(t)\,dt .
\end{equation}
When $\lambda = V$, the generalized force reduces to $-P$, and this
expression becomes the familiar work integral
\begin{equation}
W[\gamma] = \int_{\gamma} P\,dV .
\end{equation}
Each realization of the protocol therefore defines a trajectory in
thermodynamic state space carrying a fluctuating work functional.

A central result of nonequilibrium statistical mechanics is the
Jarzynski equality~\cite{Jarzynski1997},
\begin{equation}
\left\langle e^{-\beta W} \right\rangle
=
e^{-\beta \Delta F},
\end{equation}
which relates the work distribution of a driven process to the
equilibrium free--energy difference.

Within the present framework, $W[\gamma]$ may be interpreted as a
geometric action accumulated along the trajectory. For closed cycles,
or for protocols completed by reference paths on the equilibrium
manifold, the work reduces to a projected area,
\begin{equation}
W[\gamma] = \oint P\,dV
           = \iint_{\Sigma} dP \wedge dV .
\end{equation}
The Jarzynski average may then be viewed as an ensemble average over
fluctuating oriented areas in the $(P,V)$ plane,
\begin{equation}
\left\langle
\exp\!\left(
-\beta \!\iint_{\Sigma(\gamma)} dP \wedge dV
\right)
\right\rangle
=
e^{-\beta \Delta F}.
\end{equation}

Classical area laws therefore emerge as the deterministic limit of a
statistical description in which thermodynamic processes correspond to
ensembles of fluctuating trajectories. The geometric structure of
equilibrium thermodynamics extends naturally to this stochastic setting.

The same structure also gives a geometric interpretation of the
Clausius inequality. Reversible processes lie on the equilibrium
submanifold where the contact form $\alpha = dU - T\,dS + P\,dV$
vanishes. Irreversible trajectories depart from this manifold and
produce a nonzero integral of $\alpha$. Over a closed cycle this yields
\[
\oint \frac{\delta Q}{T} \le 0 ,
\]
which expresses the failure of irreversible trajectories to preserve
the underlying geometric structure.
\subsection{Thermodynamic Response and Local Cycle Geometry}
\label{sec:response_geometry}

The canonical two--form $\Omega$ defines the local work generated by
infinitesimal thermodynamic cycles. In the energy representation
$U(S,V)$, it takes the form
\begin{equation}
\Omega = -U_{SV}\, dS \wedge dV ,
\end{equation}
so that the work generated by a small reversible cycle is determined by
the mixed derivative $U_{SV}$ of the equilibrium energy surface.

This quantity admits a direct thermodynamic interpretation. By the
Maxwell relation,
\begin{equation}
U_{SV}
=
\left(\frac{\partial T}{\partial V}\right)_S ,
\end{equation}
so the local cycle geometry is controlled by the response of temperature
to volume changes at fixed entropy.

Expressing this derivative in terms of measurable response functions
gives
\begin{equation}
U_{SV}
=
\frac{T\alpha}{C_V \kappa_T},
\end{equation}
where $\alpha$ is the thermal expansion coefficient,
$\kappa_T$ the isothermal compressibility, and $C_V$ the heat capacity at
constant volume. The curvature field governing cyclic work is therefore
determined directly by thermodynamic susceptibilities.

For an infinitesimal reversible cycle spanning an oriented area
$\mathrm{Area}(\Sigma)$ in the $(S,V)$ plane, Theorem~\ref{theorem:1}
yields
\begin{equation}
W \approx -U_{SV}\,\mathrm{Area}(\Sigma).
\end{equation}
This identifies $U_{SV}$ as a local work density: thermodynamic systems
admit a field over state space whose magnitude sets the work generated
by small cycles, while the sign determines the orientation of the
response.

Regions of state space with large thermal expansion or low
compressibility correspond to enhanced curvature and therefore larger
work generation, while loci where $U_{SV}$ vanishes correspond to local
decoupling of entropy and volume variations and suppressed cycle
response.

This local geometric picture also connects naturally to nonequilibrium
fluctuation relations. When thermodynamic trajectories fluctuate due to
microscopic degrees of freedom, the work becomes a stochastic path
functional. In this setting, the Jarzynski equality relates equilibrium
free-energy differences to an exponential average over such
trajectories. Within the present framework, this average may be viewed
as sampling a distribution of geometric actions, with the classical
area law emerging as the deterministic limit. The scalar field $U_{SV}$
thus plays a dual role: it governs the local generation of work in
reversible cycles and sets the scale for geometric contributions to
work fluctuations away from equilibrium.

\section{Conclusions}

This work identifies a geometric structure underlying classical
thermodynamics. The familiar relations connecting mechanical work and
reversible heat to areas enclosed by cycles in the $(P,V)$ and $(T,S)$
planes arise as projections of a single canonical construction defined
by the First Law. On the equilibrium manifold, the contact form
$\alpha = dU - T\,dS + P\,dV$ induces a two--form whose projections
generate the classical area laws, so that the $(P,V)$ and $(T,S)$
representations are not independent, but distinct coordinate
expressions of the same intrinsic geometric object.

The central result is that thermodynamic work admits a local geometric
description. In the energy representation,
\[
\Omega = -U_{SV}\, dS \wedge dV,
\]
so that the work generated by an infinitesimal reversible cycle is
determined by the mixed derivative $U_{SV}$. This quantity admits a
direct expression in terms of measurable response functions,
\[
U_{SV} = \frac{T\alpha}{C_V \kappa_T},
\]
establishing a direct link between thermodynamic susceptibilities and
cyclic work.

Thermodynamic systems therefore admit a local work density defined over
the equilibrium manifold, and reversible cycles sample this field
through their enclosed geometry. Work is thus elevated from a global
property of cycles to a local geometric field that encodes the
work-generating capacity of thermodynamic states.

This geometric structure also connects naturally to nonequilibrium
settings. When thermodynamic trajectories fluctuate, work becomes a
stochastic path functional, and fluctuation relations such as the
Jarzynski equality may be interpreted as averages over fluctuating
geometric actions. The classical area laws then emerge as the
deterministic limit of this broader statistical description.

Beyond its conceptual significance, the present formulation provides a
practical computational framework. Given an equation of state or a
free-energy surface, the scalar field $U_{SV}$, or equivalently
$T\alpha/(C_V\kappa_T)$, defines a local cycle-response map that
identifies regions of state space where small cycles generate enhanced
work, providing a basis for the analysis and geometric optimization of
thermodynamic protocols.

Extensions of this framework to nonequilibrium and open systems, where
the effective curvature need not be positive-definite, may lead to
qualitatively new geometric effects.

\appendix
\section{Mapping a Quasistatic Cycle from $(P,V)$ to $(T,S)$}
\label{app:PV_to_TS}

The geometric discussion in the main text may be made constructive by
describing how a quasistatic thermodynamic cycle specified in the
$(P,V)$ plane is mapped into the $(T,S)$ plane.  This procedure makes
explicit how a closed trajectory in one thermodynamic chart may be
represented in another, and clarifies why the image of a generic small
cycle in the $(T,S)$ plane is not generally a rectangle.

Let a closed quasistatic cycle be parameterized by a variable
$\lambda \in [0,2\pi]$,
\begin{equation}
P = P(\lambda), \qquad V = V(\lambda),
\end{equation}
with
\begin{equation}
P(0) = P(2\pi), \qquad V(0) = V(2\pi).
\end{equation}
The temperature along the path is determined by the equation of state,
\begin{equation}
T(\lambda) = T\big(P(\lambda),V(\lambda)\big).
\end{equation}
For example, for an ideal gas one has
\begin{equation}
T(\lambda) = \frac{P(\lambda)V(\lambda)}{nR}.
\end{equation}

To construct the entropy coordinate along the cycle, one uses the
thermodynamic identity
\begin{equation}
dS
=
\left(\frac{\partial S}{\partial T}\right)_V dT
+
\left(\frac{\partial S}{\partial V}\right)_T dV .
\end{equation}
Using the standard relations
\begin{equation}
\left(\frac{\partial S}{\partial T}\right)_V = \frac{C_V}{T},
\qquad
\left(\frac{\partial S}{\partial V}\right)_T
=
\left(\frac{\partial P}{\partial T}\right)_V ,
\end{equation}
this becomes
\begin{equation}
dS = \frac{C_V}{T}\,dT
+ \left(\frac{\partial P}{\partial T}\right)_V dV .
\label{eq:app_dS_general}
\end{equation}
Hence, along the parameterized path,
\begin{equation}
\frac{dS}{d\lambda}
=
\frac{C_V}{T(\lambda)}\frac{dT}{d\lambda}
+
\left(\frac{\partial P}{\partial T}\right)_V
\frac{dV}{d\lambda}.
\label{eq:app_dS_dlambda}
\end{equation}

The entropy coordinate is therefore obtained as the cumulative integral
along the path,
\begin{equation}
S(\lambda)
=
S(0)
+
\int_0^\lambda
\left[
\frac{C_V}{T(\lambda')}\frac{dT}{d\lambda'}
+
\left(\frac{\partial P}{\partial T}\right)_V
\frac{dV}{d\lambda'}
\right] d\lambda' .
\label{eq:app_Slambda}
\end{equation}
The image of the original $(P,V)$ cycle in the $(T,S)$ plane is then
the parametric curve
\begin{equation}
\lambda \mapsto \big(T(\lambda),S(\lambda)\big).
\end{equation}

Because entropy is a state function, a reversible closed cycle satisfies
\begin{equation}
\oint dS = 0,
\end{equation}
so that
\begin{equation}
S(2\pi) = S(0).
\end{equation}
Thus the mapped trajectory closes in the $(T,S)$ plane.  The vanishing
of the full-cycle integral, however, only expresses closure of the
curve.  The shape of the image is determined by the running value
$S(\lambda)$ obtained from the partial integral in
Eq.~\eqref{eq:app_Slambda}.

This construction is completely general for quasistatic cycles once the
equation of state and the relevant thermodynamic response functions are
specified.  In particular, it shows that a generic closed curve in the
$(P,V)$ plane maps to a corresponding closed curve in the $(T,S)$ plane
whose shape depends on the local thermodynamic response of the system.
The familiar rectangular geometry arises only in the special case where
the cycle is explicitly constructed from exact isothermal and adiabatic
branches, namely the Carnot cycle.

For the local elliptical cycle considered in
Sec.~\ref{sec:IIb},
\begin{equation}
V(\lambda)=V_0+a\cos\lambda,\qquad
P(\lambda)=P_0+b\sin\lambda ,
\end{equation}
the linearized mapping derived in the main text gives
\begin{equation}
T(\lambda)\approx
T_0+\frac{P_0 a}{nR}\cos\lambda
+\frac{V_0 b}{nR}\sin\lambda ,
\end{equation}
and
\begin{equation}
S(\lambda)\approx
S_0+\frac{nC_P a}{V_0}\cos\lambda
+\frac{nC_V b}{P_0}\sin\lambda .
\end{equation}
Thus, in the local linearized approximation, the image of the
elliptical cycle in the $(P,V)$ plane is a rotated ellipse in the
$(T,S)$ plane.  This provides the local geometric interpretation of the
mapping illustrated in Fig.~\ref{fig:thermo-geometry}(b).

\vspace{0.5cm}

\section*{Data Availability Statement}
All data generated or analyzed during this study are included in this manuscript.

\begin{acknowledgments}
This work at the University of Houston was supported by the National Science Foundation under CHE-2404788 and the Robert A. Welch Foundation (E-1337).
\end{acknowledgments}

\section*{Author Contributions}
The author developed the theoretical framework, performed all derivations,
and carried out the analysis presented in this work. Generative AI tools
were used during manuscript preparation to assist with drafting,
structural organization, and verification of intermediate algebraic
steps. All theoretical formulations, results, and physical
interpretations were independently developed and validated by the author.

\section*{Conflicts of Interest}
The author declares no competing financial or non-financial interests.

\bibliographystyle{apsrev4-2}
\bibliography{refs-local}

\end{document}